\documentclass[9pt]{article}
\usepackage{spconf,amsmath,epsfig,bm,algorithm,algorithmic,float}
\usepackage{amstext,graphicx,amssymb,amsfonts,xcolor,enumerate}
\usepackage{mathrsfs,cancel}
\usepackage[thmmarks,amsmath,amsthm,hyperref]{ntheorem}
\usepackage{multirow}
\usepackage{bm,slashbox, color}
\usepackage[normalem]{ulem}
\newtheorem{theorem}{Theorem}
\newtheorem{lemma}{Lemma}

\title{Fast Binary Embedding via Circulant Downsampled Matrix -- A Data-Independent Approach}
\name{}
\name{Sung-Hsien Hsieh$^{*,**}$, Chun-Shien Lu$^{*}$, and Soo-Chang Pei$^{**}$}
\address{$^{*}$Institute of Information Science, Academia Sinica, Taipei, Taiwan\\
$^{**}$Graduate Inst. Comm. Eng., National Taiwan University, Taipei, Taiwan}
\begin{document}

\maketitle

\begin{abstract}
Binary embedding of high-dimensional data aims to produce low-dimensional binary codes while preserving discriminative power.
State-of-the-art methods often suffer from high computation and storage costs.
We present a simple and fast embedding scheme by first downsampling $N$-dimensional data into $M$-dimensional data and then multiplying the data with an $M \times M$ circulant matrix.
Our method requires $O(N+M\log M)$ computation and $O(N)$ storage costs.
%, where conventional methods need $O(NM)$ and $O(NM)$, respectively.
We prove if data have sparsity, our scheme can achieve similarity-preserving well.
Experiments further demonstrate that though our method is cost-effective and fast, it still achieves comparable performance in image applications.
% it is still comparable to data-independent binary embedding approaches in terms of performance in image applications.
\end{abstract}

%\noindent

\section{Introduction}\label{sec:introduction}
\subsection{Background and Related Work}
Embedding of high-dimensional data into low-dimensional space is an important task in diverse fields due to the concern of computation and storage costs.
In particular, embedding input data into binary space while preserving similarity is becoming popular
%, measured as the angle between two vectors in the input data space, is becoming popular
% for image retrieval, approximate nearest neighbor (ANN) search, and descriptor matching
because binary codes only require calculating Hamming distance implemented by adds.
%However, it has been noticed that for signals with dimensionality $N$, the required number of bits $M$ for saving binary codes after embedding is usually $O(N)$.
%A goal of binary embedding is to retain useful features ({\em i.e.} similarity) as Hamming distance in the binary space.

Most existing techniques can be classified into two cases: data-independent and data-dependent.
Data-independent techniques are popular due to their low-resource requirement and simplicity but often fail to give the best performance. On the contrary, data-dependent techniques often has better performance. But, along with the increase of size of data \cite{HYZ2013, JCBL2015}, they are prohibited from being applied to learning because of high computation and storage costs.

In data-independent techniques, the popular and pioneered techniques are Locality Sensitive Hashing (LSH) \cite{Charikar2002} and its extension Shift-Invariant Locality Sensitive Hashing (SKLSH) \cite{Raginsky2009} wherein embedding is based on random projection to achieve similarity-preserving.
In \cite{KL2009}, dimensionality reduction inherent in compressive sensing is exploited via random projection for image hash design.
Gong {\em et al.}\cite{Gong2013} proposed a bilinear projection to further reduce computation and storage overheads during embedding. Chang {\em et al.} \cite{ChangeFu2014}  proposed using a circulant matrix for projecting data because projection can be speeded up by Fast Fourier Transform (FFT).
A learning mechanism is also considered in \cite{Gong2013}\cite{ChangeFu2014}.

As for data-dependent techniques, different optimization criteria are used in the learning phase.
For example, Iterative Quantization (ITQ) \cite{Gong2013-ITQ} aims to minimize quantization error after PCA. \cite{Xia2015} proposed a sparsity regularizer in learning to reduce computation cost.
Recently, deep neural network (DNN) \cite{Lai2015} is used to jointly learn features and binary codes simultaneously.
These methods learn compact codes especially for low-dimensional embedding. But, most of them require $O(N^2)$ computation and storage costs that may not be practical.
Online learning is another issue along with increase of data \cite{HYZ2013}\cite{JCBL2015}.

\subsection{Contributions of This Paper}
In this paper, we propose a data-independent approach, including two steps: downsampling $N$-dimensional data into $M$-dimensional data first and then multiplying the data with an $M \times M$ circulant matrix.
The proposed method, achieving $O(N+M\log M)$ in computation cost and $O(N)$ in storage cost, obviously outperforms state-of-the-art methods.
%Compared with \cite{Charikar2002}\cite{Gong2013}, the proposed scheme uses circulant matrix to reduce computation.
Although our method and \cite{ChangeFu2014} are conceptually similar by introducing a circulant matrix for binary embedding, the major differences include:
(i) We use downsampling matrix to compress the signal first, leading to the fact that the size of our circulant matrix can depend on $M$ only instead of $N$.
In \cite{ChangeFu2014}, whatever $M$ is, it  requires the same computation cost $O(N \log N)$ because of using FFT for speeding computation.
Thus, when $M \ll N$, \cite{Charikar2002}\cite{Gong2013} are even faster than \cite{ChangeFu2014}.
(ii) We theoretically prove that even though downsampling is used, by combining downsampling with randomization,  similarity-preserving is still satisfied well.

In addition to the fact that the computation and storage costs of our method are smaller than those of previous methods, experimental results reveal that their performances in image applications are comparable.

\section{Notations}\label{sec:notations}
We display a matrix or a vector as bold. Let $\bm{V}$ be a matrix, where $\bm{v}_{i}$ is the $i^{th}$ column of $V$ and $\bm{v}^{j}$ is the $j^{th}$ row of $V$. $(\bm{V})_{i,j}$ is the $\left( i,j \right)^{th}$ entry of $\bm{V}$.
Let $\bm{u} \in \mathbb{R}^{N}$ be a vector and let $circ(\bm{u})$ be a circulant matrix generated based on the seed vector $\bm{u}$.
For example, for $\bm{U}=circ(\bm{u})$, the first row is $[(\bm{u})_{0},(\bm{u})_{1},...,(\bm{u})_{N-1} ]$, the second row is $[(\bm{u})_{N-1},(\bm{u})_{0},...,(\bm{u})_{N-2} ]$, and the last row is $[(\bm{u})_{1},(\bm{u})_{2},...,(\bm{u})_{0} ]$.

\section{Proposed Method}\label{sec:proposed method}
We first describe how to design a data-independent projection matrix to achieve both the lowest computation and storage costs in the literature.
Then, we prove that the proposed method still satisfies similarity-preserving property.
In this paper, following \cite{Charikar2002}\cite{ChangeFu2014}, similarity is measured as the
angle between two vectors in the input data space.

\subsection{Construction of Projection Matrix}\label{sec:construction}
The core idea is to design a projection matrix composed of a downsampling matrix and a circulant matrix achieving: (i) $O(N + M \log M)$ operations for fast embedding process. (ii) $O(N)$ bits for saving the projection matrix. (iii) Angle-preserving after embedding.

Binary embedding or $1$-bit compressive sensing \cite{BB2008} is defined as:
\begin{equation}
\label{eq: cbe}
\displaystyle \bm{h}=sign(\bm{A}\bm{x}).
\end{equation}
where $\bm{x} \in \mathbb{R}^{N}$ is an input signal, $\bm{h}\in \mathbb{R}^{M}$ is the corresponding binary code, $sign(\cdot)$ is a sign function, and $\bm{A}\in \mathbb{R}^{M \times N} $ is a projection matrix defined as:
%where $\bm{x} \in \mathbb{R}^{N}$ is an input data and $\bm{A}$ is a projection matrix defined as:
\begin{equation}
\label{eq: A}
\bm{A}=\bm{D} \Phi \bm{R}.
\end{equation}

Specifically, $\bm{R}$ is either a uniform random permutation matrix (global randomizer) or a diagonal random matrix (local randomizer) whose diagonal entries $(\bm{R})_{i,i}$ are i.i.d Bernoulli random variables with equal probability.
In our paper, $\bm{R}$ implements both global randomizer and local randomizer simultaneously\footnote{Specifically, let $\bm{R}_{1}$ be a global randomizer and let $\bm{R}_{2}$ be a local randomizer. Then, $\bm{R}=\bm{R}_{1}\bm{R}_{2}$.}.
$\Phi \in \mathbb{R}^{M \times N}$ is a downsampling matrix with $(\bm{\Phi})_{i,j}=1$ if $(-i+j) \text{ mod } M =0 $ for $0\leq i,j\leq N-1$.
$\bm{D} = circ(\bm{d}^{0}) \in \mathbb{R}^{M \times M}$ is a circulant matrix with seed vector $\bm{d}^{0}$, where $\bm{d}^{j}$ is the $j^{th}$ row of $\bm{D}$, to achieve:
1) faster computation than traditional random matrix; 2) fairly spreading the information into each bit.

Based on Eq. (\ref{eq: A}), the computation cost includes
(i) $\bm{D}$ is implemented by FFT with $O( M \log M)$.
(ii) $\bm{\Phi}$, in fact, acts to downsample $\bm{R}\bm{x}$ and cost $O(N)$ adds and zero multiplications. (iii) Each column in $\bm{R}$ only has a non-zero entry with either $-1$ or $+1$ and $\bm{R}$ costs $O(N)$ adds.
In sum, the computation cost is $O(N+M \log M)$.

Furthermore, in terms of storage cost, $\bm{D}$ is equivalent to $circ(\bm{d_0})$ and saving $\bm{d_0}$ costs $O(M)$.
$\bm{\Phi}$ is not necessary to be saved since $\bm{\Phi}$, in fact, is finished by:
\begin{equation}
\label{eq: Phi x operator_ver}
\displaystyle (\bm{\Phi}\bm{Rx})_{k}= \sum_{i=0}^{\frac{N}{M}-1} (\bm{Rx})_{k+iM}.
\end{equation}
$\bm{R}$ only has $N$ non-zero entries and costs $O(N)$.
Thus, the total storage cost is $O(N)$.

Table \ref{table: computation cost and storage cost} depicts the comparison between our scheme and representative fast embedding methods.
Specifically, $\bm{A}$'s in \cite{Charikar2002} and \cite{ChangeFu2014} are designed as a Gaussian random matrix and circulant matrix, respectively.
\cite{Gong2013} reshapes $\bm{x}$ into two-dimensional data, which are projected by two separable Gaussian random matrices with smaller size.

Our approach exhibits the best desired requirement in terms of computation and storage costs.
In addition, when one only focuses on the number of multiplications (adds can be handled more efficiently than multiplications) \cite{Hassanieh2012acm}, our scheme only requires $O(M \log M)$ computation cost.

\begin{table}[!htbp]
\caption{Comparison of computation and storage costs for data-independent binary embedding methods.}
\begin{tabular}{|c|c|c|}
\hline
Methods & Computation & Storage  \\
\hline\hline
Full projection \cite{Charikar2002} & $O(MN)$ & $O( MN )$  \\
Bilinear proj. \cite{Gong2013}& $O(N^{1.5})$ & $O( N )$  \\
Circulant proj. \cite{ChangeFu2014} & $O(N \log N)$ & $O( N )$  \\
Our scheme & $O(N + M \log M)$ & $O( N )$  \\
\hline
\end{tabular}
\label{table: computation cost and storage cost}
\end{table}

%To benefit from aforementioned advantages, the matrix , in fact, is not necessary to be designed as circulant matrix. For example, $(\bm{A}^{j})_{i}=1$ for $i  =(j-1)\frac{N}{M}+1$ to $j\frac{N}{M}$. Otherwise, $(\bm{A}^{j})_{i}=0$. However, D. Valsesia and E. Magli \cite{Valsesia2012} show that if $\bm{\Phi}$  is circulant, it is potential to do post-processing after embedding based on Theorem \ref{thm: circulant matrix exchange}.

%\begin{theorem} \label{thm: circulant matrix exchange} (slightly revised from \cite{Valsesia2012}) \\
%Let $\bm{T}=circ([\bm{t}\ 0_{1 \times (N-K)}]) \in \mathbb{R}^{N \times N}$ be a circulant matrix, where $\bm{t}$ is a $1 \times K$ ($K \leq M$) non-zero vector.
%Let $\hat{\bm{T}} \in \mathbb{R}^{M \times M} = circ([\bm{t}\ 0_{1 \times (M-K)}])$, let $\Phi$ be a $M \times N$ partial circulant matrix, and let $\bm{y}=\Phi \bm{x}$ and $\Phi \bm{T x}$ denote the measurements of $\bm{x}$ and measurements of filtered signal ($\bm{Tx}$), respectively.
%Then,
%$$ \left( \hat{\bm{T}}\bm{y} \right)_{i} = \left( \Phi \bm{T x} \right)_{i}\ \ \text{if and only if} \ \ i \in [0,\ M-K ]. $$
%\end{theorem}
%$\bm{t}$ serves as a filter. Thus, Theorem \ref{thm: circulant matrix exchange} illustrates that instead of filtering $\bm{x}$ in high dimension, filtering $\bm{y}$ in low dimension achieves the same result, while saving the computation cost. In this case, extra $O(M \log M)$ operations are required when $\hat{\bm{T}}\bm{y}$ is computed by Fast Fourier Transform (FFT).

\subsection{Angle-Preserving Property Based on Sparsity}\label{sec:CBE}

Like \cite{Charikar2002}\cite{ChangeFu2014}\cite{Kim2015},  we analyze the property of similarity (angle)-preserving for the proposed scheme in this section.
Angle-preserving is useful because angle includes the information about similarity between data, which is an important physical property in many applications, including image retrieval and  nearest neighbor search.
%In fact, data-independent techniques such as \cite{Charikar2002} \cite{ChangeFu2014} often focuses on this property.

Suppose $\mathcal{H}_{M}\left(\bm{x}_{1},\bm{x}_{2}\right)$ is the normalized Hamming distance between $\bm{x}_{1}$, $\bm{x}_{2}$:
\begin{equation}
\label{eq: hamming distance}
\displaystyle \mathcal{H}_{M}\left(\bm{x}_{1},\bm{x}_{2}\right)=\frac{1}{2M}\sum_{i=0}^{M}\left| sign(\bm{a}^{i}\bm{x}_{1}) - sign(\bm{a}^{i}\bm{x}_{2}) \right|,
\end{equation}
where $\bm{a}^{j}$ is the $j^{th}$ row of $\bm{A}$.
It is expected that $\mathcal{H}_{M}\left(\bm{x}_{1},\bm{x}_{2}\right)$ is related to the angle  $\theta$ between $\bm{x}_{1}$ and $\bm{x}_{2}$.
The ideal case of angle-preserving property satisfies $E\left\{ \mathcal{H}_{M}\left(\bm{x}_{1},\bm{x}_{2}\right) \right\} = c\theta$, where $c$ is a constant, and $Var\left\{ \mathcal{H}_{M}\left(\bm{x}_{1},\bm{x}_{2}\right) \right\} = 0$.

If $\bm{A}$ is drawn from i.i.d distribution, which collides with the proposed method, M. S. Charikar \cite{Charikar2002} has shown $E\left\{ \mathcal{H}_{M}\left(\bm{x}_{1},\bm{x}_{2}\right) \right\} = \frac{\theta}{\pi}$ and $Var\left\{ \mathcal{H}_{M}\left(\bm{x}_{1},\bm{x}_{2}\right) \right\} = \frac{\theta(\pi-\theta)}{M\pi^2}$.
Chang {\em et al.} \cite{ChangeFu2014} only show by experiments if $\bm{A}$ is a circulant matrix, whose first row is a Gaussian random vector, the \textit{sample} mean and \textit{sample} variance of $\mathcal{H}_{M}\left(\bm{x}_{1},\bm{x}_{2}\right)$ corresponding to $\bm{A}$  approximates the results of M. S. Charikar \cite{Charikar2002}.

Our proof of angle-preserving property includes two steps:
1) Let $\hat{\bm{x}}=\bm{\Phi}\bm{R}\bm{x}$. We prove $\bm{D}$ can preserve the angle $\hat{\theta}$ between $\hat{\bm{x}}_{1}$ and $\hat{\bm{x}}_{2}$.
2) Then, we show $\hat{\theta} \sim \theta$ holds, which implies our scheme preserves $\theta$.

For the first step,  \cite{ChangeFu2014} has validated if $\bm{d}^{0}$ is a Gaussian random vector, then $\bm{D}$ preserves  $\hat{\theta}$ between $\hat{\bm{x}}_{1}$ and $\hat{\bm{x}}_{2}$ after embedding.
For the second step, Chang and Wu \cite{Chang2013} show that if a matrix satisfies $\delta_{K}$-RIP, it also preserves angle with the distortion being proportional to $\delta_{K}$ after embedding.
\begin{theorem} \label{thm: RIP}($\delta_{K}$-RIP  \cite{Baraniuk07} )
Let $\bm{A} \in \mathbb{R}^{M \times N}$ be a random matrix drawn according to any distribution that satisfies the concentration inequality.
Then, for any $K$-sparse signal $\bm{x}$ and any $0 < \delta_{K} < 1 $, we have
\begin{equation}
\label{eq: RIP for random matrix}
\displaystyle (1-\delta_{K})\| \bm{x} \|_{2}^{2} \leq \| \bm{A}\bm{x} \|_{2}^{2} \leq (1+\delta_{K})\| \bm{x} \|_{2}^{2},
\end{equation}
with the probability
$$ \geq 1- 2 \left(\frac{12}{\delta_{K}}\right)^{K}e^{-\left(\frac{\delta_{K}^{2}}{16}-\frac{\delta_{K}^{3}}{48} \right)M }.$$
\end{theorem}

We call a matrix satisfying $\delta_{K}$-RIP when Eq. (\ref{eq: RIP for random matrix}) holds. In other words, if $\bm{\Phi}\bm{R}$ satisfies $\delta_{K}$-RIP, $\bm{\Phi}\bm{R}$ preserves the angle.
To date, finding a deterministic matrix satisfying RIP within polynomial time, however, is still an open problem \cite{Foucart2013}.
Unfortunately, the proposed projection matrix is deterministic to violate Theorem \ref{thm: RIP}.

To overcome this problem, we derive another theoretical bound about $\delta_{K}$ along with the lower bound of the probability. We start from the following Lemma.
 \begin{lemma}
Let $\bm{x} \in \mathbb{R}^{N}$ be $K$-sparse, let $ (\hat{\bm{x}})_{k} = \sum_{i \in S_{k}} (\bm{R}\bm{x})_{i}$ with $k = 0,1,...,M-1$, let $S_{k}=\{ k+jM |(\bm{Rx})_{k+jM} \neq 0 \text{ for } j=0,...,\frac{N}{M}-1 \}$, and let $\kappa_{k} =|S_{k}|$.
Then,
 $$\left| \left\{ k |  \kappa_{k} \geq 2 \text{ for } k=0,...,M-1 \right\} \right| < f , $$
 hold for $f=1,...,\frac{K}{2}$ with the probability being larger than $$\displaystyle \geq 1 - \dbinom{M}{f}\dbinom{\frac{N}{M}}{2}^{f}\dbinom{N-2f}{K-2f}/\dbinom{N}{K}.$$
 Moreover, by Stirling's formula, the bound is relaxed into
$$ \geq 1 - \frac{1}{\sqrt{2\pi f}}(\frac{eK^2}{2Mf})^{f}.$$
 \label{lemma: the number of aliasing support}
 \end{lemma}
 \begin{proof}
To simplify the notation, let $E_{f}$ be the event with  $| \{ k |  \kappa_{k} \geq 2 \text{ for } k=0,...,M-1 \} | < f$.
The event is related to the positions of non-zero entries of $\bm{R}\bm{x}$ but is unrelated to their values. Since $\bm{R}$ permutes $\bm{x}$ randomly, the positions of non-zero entries of $\bm{R}\bm{x}$ are uniformly distributed.
Thus, $P\{E_{f}\}$ is considered as a combination problem.
Let $\dbinom{N}{K}$ be all combinations of $N$ positions taking $K$ positions being non-zeros at a time.
Then, $P\{E_{f}\}$ is equal to divide the number of combinations belonging to $E_{f}$ by $\dbinom{N}{K}$.

Instead of calculating $P\{E_{f}\}$ directly, we focus on $P\{E_{f}^{c}\}$, which is the complement of $E_{f}$.
Specifically, $E_{f}^{c}$ is the event with $| \{ k |  \kappa_{k} \geq 2 \text{ for } k=0,...,M-1 \} | \geq f $.
Then,$$ P\{E_{f}^{c}\} \leq   \dbinom{M}{f}\dbinom{\frac{N}{M}}{2}^{f}\dbinom{N-2f}{K-2f}/\dbinom{N}{K}.$$
$\dbinom{M}{f}\dbinom{\frac{N}{M}}{2}^{f}$ means choosing $f$ sets from $S_{0},S_{1},...,S_{M-1}$ such that the chosen sets satisfy $\kappa_{k} = 2$.
 Thus, $2f$ non-zero entries of $\bm{x}$ are arranged.
 Then, $\dbinom{N-2f}{K-2f}$ means the remaining ($K-2f$) non-zero entries of $\bm{x}$ distribute randomly among the remaining $N-2f$ positions.

Consequently, since $P\{E_{f}\}+P\{E_{f}^{c}\} =1$, we have $P\{E_{f}\} = 1 - P\{E_{f}^{c}\}  \geq  1 -  \dbinom{M}{f}\dbinom{\frac{N}{M}}{2}^{f}\dbinom{N-2f}{K-2f}/\dbinom{N}{K}$.

Further, the term $\dbinom{M}{f}\dbinom{\frac{N}{M}}{2}^{f}\dbinom{N-2f}{K-2f}/\dbinom{N}{K} $ is  approximated by:
$$
\begin{aligned}
    & \dbinom{M}{f}\dbinom{\frac{N}{M}}{2}^{f}\dbinom{N-2f}{K-2f}/\dbinom{N}{K} \\
    &\leq  \frac{N^{2f}}{f!(2M)^f}\frac{ (N-2f)! K! }{(K-2f)! N!} \\
    &\leq  \frac{K^{2f}}{f!(2M)^f}\frac{ (N-2f)! N^{2f} }{ N!} \\
    &\sim  \frac{1}{\sqrt{2\pi f}}(\frac{eK^2}{2Mf})^{f}.
\end{aligned}
$$
The last deviation is due to $f! \sim \sqrt{2\pi f}(\frac{f}{e})^{f} $ by Stirling's formula, where the approximation is more accurate when $N$ is large enough.
Thus, $P\{E_{f}\} \geq 1 -  \frac{1}{\sqrt{2\pi f}}(\frac{eK^2}{2Mf})^{f}$.  We complete this proof.
\end{proof}

It should be noted that, if $\kappa_{k} = 1$, it implies $(\hat{\bm{x}})_{k}$ is equal to one of non-zero entries of $\bm{Rx}$.
If $\kappa_{k} = 1$ or $0$ for $0\leq k\leq M-1$, it means no distance distortion and $\| \hat{\bm{x}} \|_{2} = \| \bm{Rx} \|_{2} = \| \bm{x} \|_{2}$.
Thus, based on Lemma \ref{lemma: the number of aliasing support}, we can derive in Theorem \ref{thm: RIP for no distortion} the probability with $\delta_{K}=0$.

\begin{theorem} \label{thm: RIP for no distortion}
Let $\bm{\Phi R} \in \mathbb{R}^{M \times N}$. Then, for any $K$-sparse $\bm{x}$, we have $\delta_{K}=0$ such that
\begin{equation}
\label{eq: RIP for no distortion}
 \| \bm{\Phi R}\bm{x} \|_{2}^{2} = \| \bm{x} \|_{2}^{2},
\end{equation}
with the probability
$$\displaystyle \geq 1 - M\dbinom{\frac{N}{M}}{2}\dbinom{N-2}{K-2}/\dbinom{N}{K}.$$
Moreover, by Stirling's formula, the bound is relaxed into
$$ \geq 1 - \frac{1}{\sqrt{2\pi }}\left( \frac{eK^2}{2M} \right).$$
\end{theorem}

\begin{proof}
Following the same notations in the proof of Lemma \ref{lemma: the number of aliasing support}, let $E_{1}$ be the event with $| \{ k |  \kappa_{k} \geq 2 \text{ for } k=0,...,M-1 \} | < 1$.
  In other words, if $E_{1}$ occurs, it means $\kappa_{k} = 1$ or $0$ for $0\leq k\leq M-1$.
  Thus, $P\{ \| \bm{\Phi R}\bm{x} \|_{2}^{2} = \| \bm{x} \|_{2}^{2} \} = P\{ E_{1} \}$.
  $P\{ E_{1} \}  \geq 1 - M\dbinom{\frac{N}{M}}{2}\dbinom{N-2}{K-2}/\dbinom{N}{K}$ is calculated by setting $f=1$ in Lemma \ref{lemma: the number of aliasing support}. We complete this proof.
 \end{proof}
Theorem \ref{thm: RIP for no distortion} indicates that, if $K \leq O(\sqrt{M})$, the probability of $\delta_{K}=0$ is high enough.
We will validate  Theorem \ref{thm: RIP for no distortion} by experiments later.

We further extend Theorem \ref{thm: RIP for no distortion} to consider different values of $\delta_K$.
Nevertheless, if $\kappa_{k} > 1$, $(\bm{\Phi}\bm{R} \bm{x})_{k}$ is the sum of at least two non-zero entries of $\bm{x}$.
In this case, different signals ($\bm{x}$'s) will led to different distance distortions.
To simplify the problem, we assume $\bm{x} \in \left\{ 0,1 \right\}^{N}$.
Under the circumstance, theoretical bound for $\delta_{K}$ is derived in Theorem \ref{thm: RIP for ce}.

\begin{theorem} \label{thm: RIP for ce}
Let $\bm{\Phi R} \in \mathbb{R}^{M \times N}$ and $g=\dbinom{\frac{N}{M}}{2}$. Then, for any $K$-sparse $\bm{x} \in \left\{ 0,1 \right\}^{N}$ and any $\delta_{K} \in \left\{0,\frac{2g}{K},\frac{4g}{K}...,g-\frac{2g}{K} \right\}$, we have
\begin{equation}
\label{eq: RIP for ce}
\displaystyle (1-\delta_{K})\| \bm{x} \|_{2}^{2} \leq \| \bm{\Phi R}\bm{x} \|_{2}^{2} \leq (1+\delta_{K})\| \bm{x} \|_{2}^{2},
\end{equation}
 with the probability
  $$\displaystyle \geq 1 - \dbinom{M}{f}\dbinom{\frac{N}{M}}{2}^{f}\dbinom{N-2f}{K-2f}/\dbinom{N}{K},$$
 where $f=\frac{\delta_{K}K}{2g}+1$. Moreover, by Stirling's formula, the bound is relaxed into
$$ \geq 1 - \frac{1}{\sqrt{2\pi f}}(\frac{eK^2}{2Mf})^{f}.$$
\end{theorem}
\begin{proof}
We use the same notation and definition in Lemma \ref{lemma: the number of aliasing support}. If $E_{f}$ occurs, without loss of generality, let $\kappa_{i} \geq 2 $ for $i=0,...,f-2$ and $\kappa_{j} = 1$ for $j=f-1,...,M-1$. Then,
\begin{equation}
\begin{aligned}
          &\|\bm{\Phi R}\bm{x} \|_2^2 \leq (1+\delta_{K})\| \bm{x} \|_2^2 \\
          &\Rightarrow  \delta_{K} = \max_{\bm{x}} \frac{ \|\bm{\Phi R}\bm{x} \|_2^2 - \| \bm{x} \|_2^2 }{\| \bm{x} \|_2^2}  \\
          &\Rightarrow  \delta_{K} = \max_{\bm{x}} \frac{ \sum_{k=0}^{f-2}\sum_{i \in S_{k}, j \in S_{k}/i } 2(\bm{Rx})_{i}(\bm{Rx})_{j}  }{\| \bm{x} \|_2^2} \\
          &\Rightarrow  \delta_{K} = \frac{2g(f-1)}{K}
\end{aligned}
\label{eq: rip deviation}
\end{equation}
The derivation in last line of Eq. (\ref{eq: rip deviation}) comes from the fact that the non-zero entries of $\bm{x}$ are $1$. Thus, $2(\bm{Rx})_{i}(\bm{Rx})_{j}$ has maximal value $2$.
 Further, the cardinality of $\{ (i,j) | i \in S_{k}, j \in S_{k}/i \}$ is $\dbinom{\kappa_{k}}{2}$. In the worst case, $\kappa_{k}=\frac{N}{M}$. Thus,
 $$ \sum_{i \in S_{k}, j \in S_{k}/i } 2(\bm{Rx})_{i}(\bm{Rx})_{j} = \sum_{i \in S_{k}, j \in S_{k}/i } 2 = 2\dbinom{\frac{N}{M}}{2}=2g.$$

Consequently, $\frac{2g(f-1)}{K} =  \delta_{K}$ or $f=\frac{\delta_{K}K}{2g}+1$. If $f=1$. it implies that $0 = \delta_{K}$ with the probability $P\{ E_{1}\}$, that is a special case like Theorem \ref{thm: RIP for no distortion}.
Since $f \in \left\{ 1,2,...,\frac{K}{2} \right\}$, we have  $\delta_{K} \in \left\{0,\frac{2g}{K},\frac{4g}{K},...,g-\frac{2g}{K}\right\}$ along with the corresponding probability $P\{ E_{f}\} = P\{ E_{\frac{\delta_{K}K}{2g}+1}\}$.   We complete this proof.
\end{proof}

We want to briefly discuss why we assume $\bm{x} \in \left\{ 0,1 \right\}^{N}$ instead of other signal types such as Gaussian random signal. The larger $\sum_{k=0}^{f-1}\sum_{i \in S_{k}, j \in S_{k}/i } 2(\bm{Rx})_{i}(\bm{Rx})_{j}$ is, the large $\delta_{K}$ is. Thus, assuming $\bm{x}$ has constant energy such that $\|\bm{x}\|_{2} = c$, the largest $\delta_{K}$ is equivalent to solving the following optimization problem:
\begin{equation}
\begin{aligned}
        \max_{\bm{x}} &\sum_{k=0}^{f-1}\sum_{i \in S_{k}, j \in S_{k}/i } 2(\bm{Rx})_{i}(\bm{Rx})_{j} \\
            &\text{subject to } \|\bm{x}\|_{2} = c.
\end{aligned}
\label{eq: optimization}
\end{equation}
By solving the optimization problem by Lagrange multiplier, the optimal value is achieved with the constraint that $(\bm{Rx})_{i}=(\bm{Rx})_{j}$ with $i,j \in S_{k}$ for $k=0,...,f-1$. If $\bm{R}$ is a deterministic matrix, it is easy to obtain optimal solution $\bm{x}$. However, $\bm{R}$ is a randomizer resulting in random locations and random sign of $\bm{x}$. By assuming $\bm{x} \in \left\{ 0,1 \right\}^{N}$, $(\bm{Rx})_{i}=(\bm{Rx})_{j}$ holds with high probability. We emphasize that rigorous proof is still absent and should be discussed in the future work.

%it should be noted that the proof is based on the assumption of random locations of non-zero entries. Nature signals maybe lack randomness. Thus, be letting $\bm{R}$ be the global randomizer to forces $\bm{Rx}$ to meet the assumption. In addition, the angle between $\bm{R}\bm{x}_{1}$ and $\bm{R}\bm{x}_{2}$ is still the same with that between $\bm{x}_{1}$ and $\bm{x}_{2}$. This is the main concern to assume $\bm{x} \in \left\{ 0,1 \right\}^{N}$.

%In addition, Theorem \ref{thm: RIP for ce} requires the assumption $\bm{x} \in \left\{ -1,0,1 \right\}^{N}$, but Lemma \ref{lemma: the number of aliasing support} holds for all $K$-sparse signals. If $f$

To check whether $\Phi$ is good enough to satisfy $\delta_{K}$-RIP from empirical and theoretical results, we compare it  with Gaussian random matrix, which is admitted to be a good choice for satisfying $\delta_{K}$-RIP.
Let $\bm{A}$ be designed as either a Gaussian random matrix drawn from $\mathcal{N}\left( 0,\frac{1}{M} \right)$ or the proposed projection matrix.
A Monte Carlo method is used to estimate RIP.
By generating a set of $K$-sparse signals ({\em i.e.,} $\bm{x}$'s),  where non-zero entries are $1$'s,
$E\left\{ \delta_{K} \right\}$ can be estimated.
Table \ref{table: RIC for different matrices} shows the empirical results, where each one is obtained from the mean of $100,000$ trials.
The proposed matrix benefits from the sparsity property and outperforms Gaussian random matrix.
%The simulation results actually meet the theoretical prediction.
Basically, the simulation results actually meet the theoretical prediction.
Moreover, Table \ref{table: RIC for different matrices with gaussian signal} shows the case that non-zero entries of $\bm{x}$ are drawn from $\mathcal{N}(0,1)$.
We can see that $\delta_{K}$'s are smaller than those in Table \ref{table: RIC for different matrices}.

%However, the drawbacks of the proposed method is that the reconstruction of $\bm{x}$ given $\bm{y}$ and $\bm{A}$, which is the hot issue in compressive sensing (CS) \cite{Candes2005}, is impossible because the mutual incoherence between columns of $\bm{x}$ is 1. But, embedding does not consider the reconstruction.

In addition, the lower bound of probability of satisfying $\delta_{K}$-RIP in Theorem \ref{thm: RIP for ce} is tighter than that in Theorem \ref{thm: RIP}, as shown in Fig. \ref{fig:RIP Estimation}, where solid curves denote the empirical results generated by Monte Carlo method and dash curves denote the corresponding theoretical lower bounds based on Theorem \ref{thm: RIP} and  Theorem \ref{thm: RIP for ce}.
Fig. \ref{fig:RIP Estimation} reveals that the lower bound in Theorem \ref{thm: RIP} is not trivial only when $\bm{x}$ is very sparse.
Otherwise, it is always zero.
Fig. \ref{fig:RIP Histogram} shows the histogram of $\delta_{K}$ under different settings of $N$, $M$, and $K$.
The horizontal axis in Fig. \ref{fig:RIP Histogram}(b) is discrete because of $\bm{x} \in \left\{ 0,1 \right\}^{N}$.
In sum, the proposed projection matrix has a higher probability to satisfy $\delta_{K}$-RIP with small $\delta_{K}$.

Consequently, since our designed $\bm{\Phi R}$ can satisfy $\delta_{K}$-RIP, it also preserves similarity between two data, as proved in \cite{Chang2013}.
Combined with the fact that $\bm{D}=circ(\bm{d}^0)$, with $\bm{d}^0$ being a Gaussian random vector, also preserves the angle between two data \cite{ChangeFu2014}, our proposed $\bm{D \Phi R}$ still retains angle-preserving property.

\begin{figure}[h]
\begin{minipage}[b]{.5\linewidth}
  \centering{\epsfig{figure=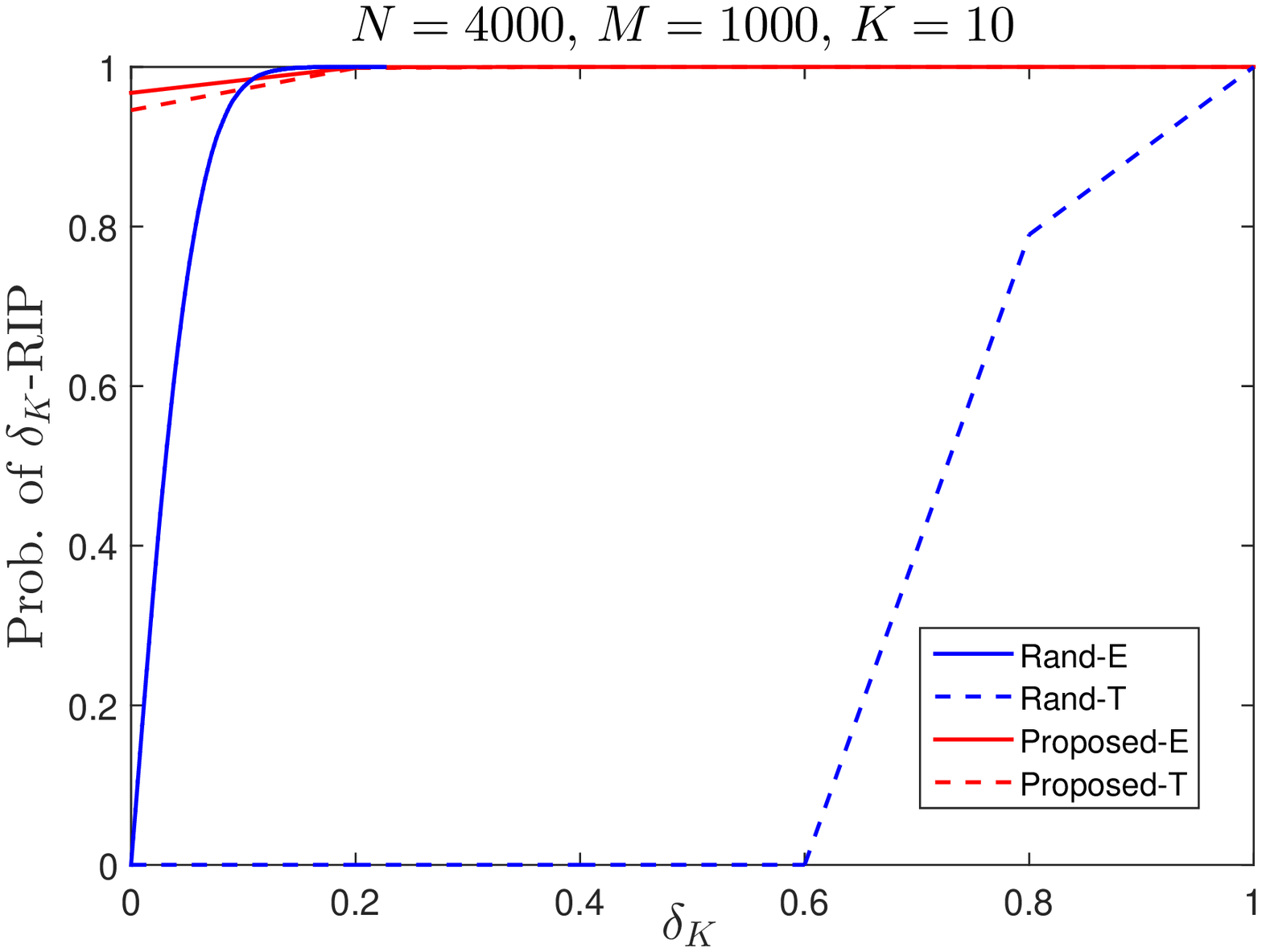,width=1.7in}}
  \centerline{(a)}
\end{minipage}
\begin{minipage}[b]{.48\linewidth}
  \centering{\epsfig{figure=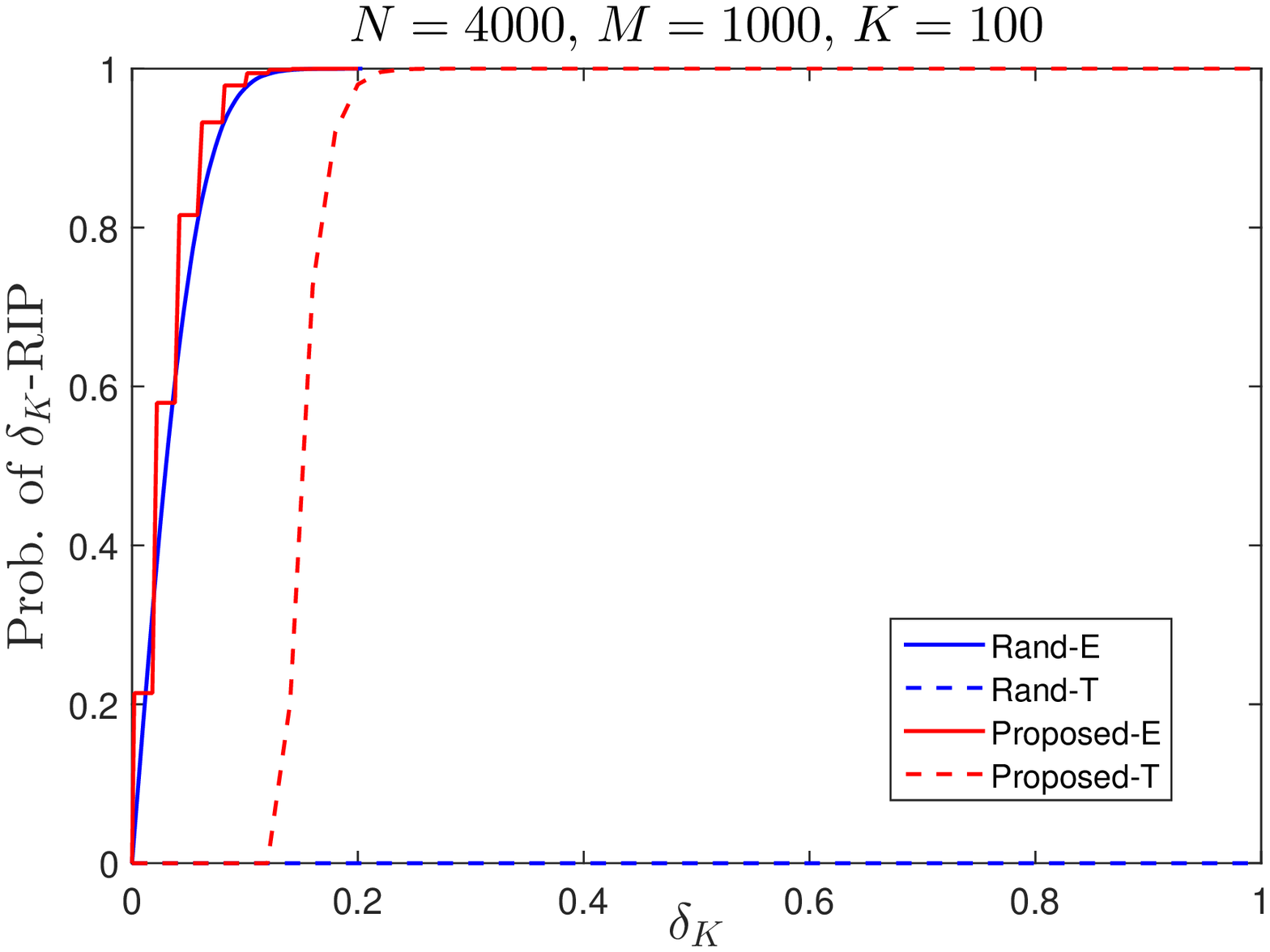,width=1.7in}}
  \centerline{(b)}
\end{minipage}
\caption{Probability of satisfying $\delta_{K}$-RIP versus $\delta_{K}$ when $\bm{A}$ is either the proposed matrix or a Gaussian random matrix.
Proposed-E and Rand-E denote empirical results while Proposed-T and Rand-T denote the lower bounds of probability in Theorem \ref{thm: RIP for ce} and Theorem \ref{thm: RIP}, respectively.
(a)$N=4000$, $M=1000$, $K=10$. (b)$N=4000$, $M=1000$, $K=100$.}
\label{fig:RIP Estimation}
\end{figure}

\begin{figure}[h]
\begin{minipage}[b]{.5\linewidth}
  \centering{\epsfig{figure=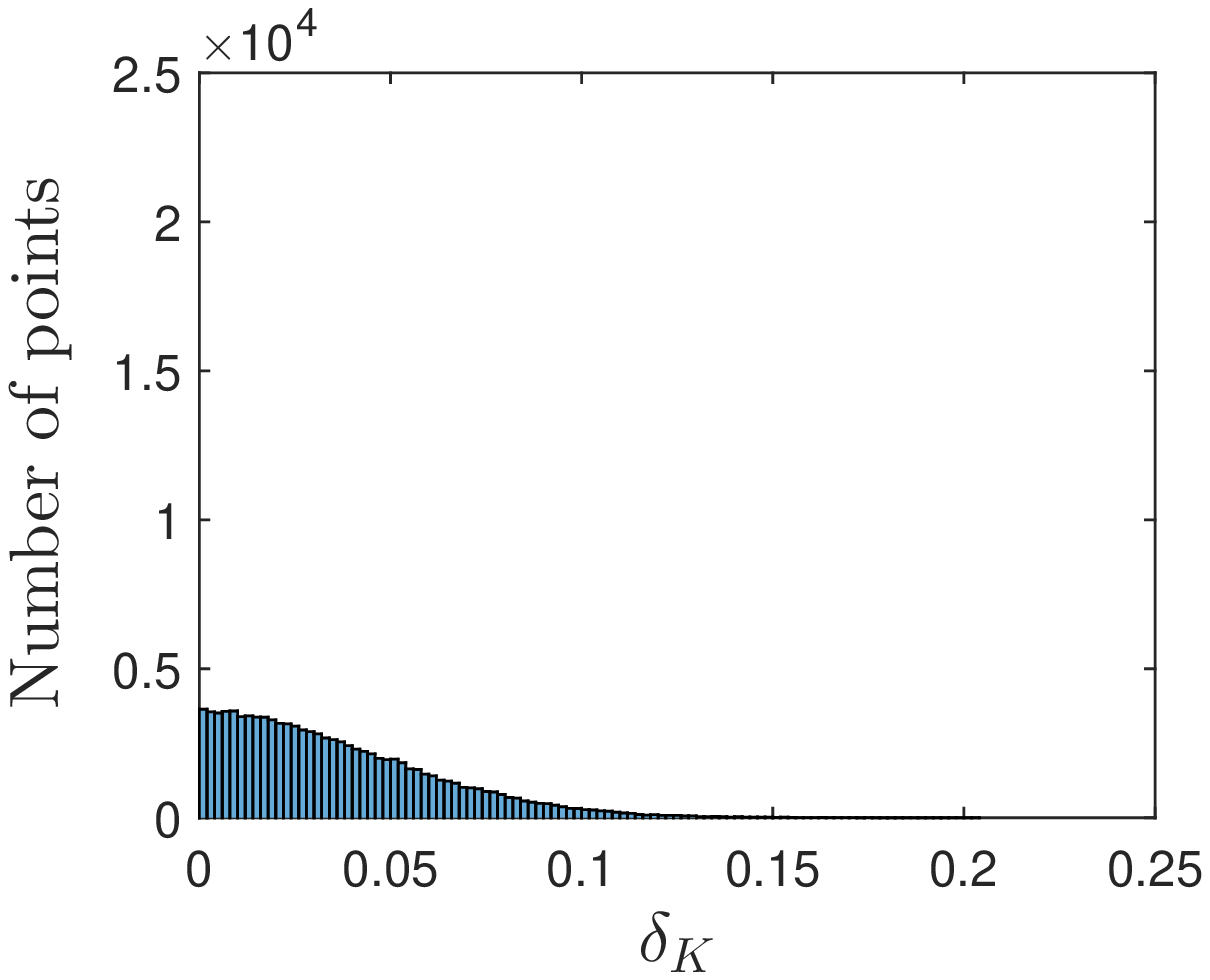,width=1.7in}}
  \centerline{(a)}
\end{minipage}
\begin{minipage}[b]{.48\linewidth}
  \centering{\epsfig{figure=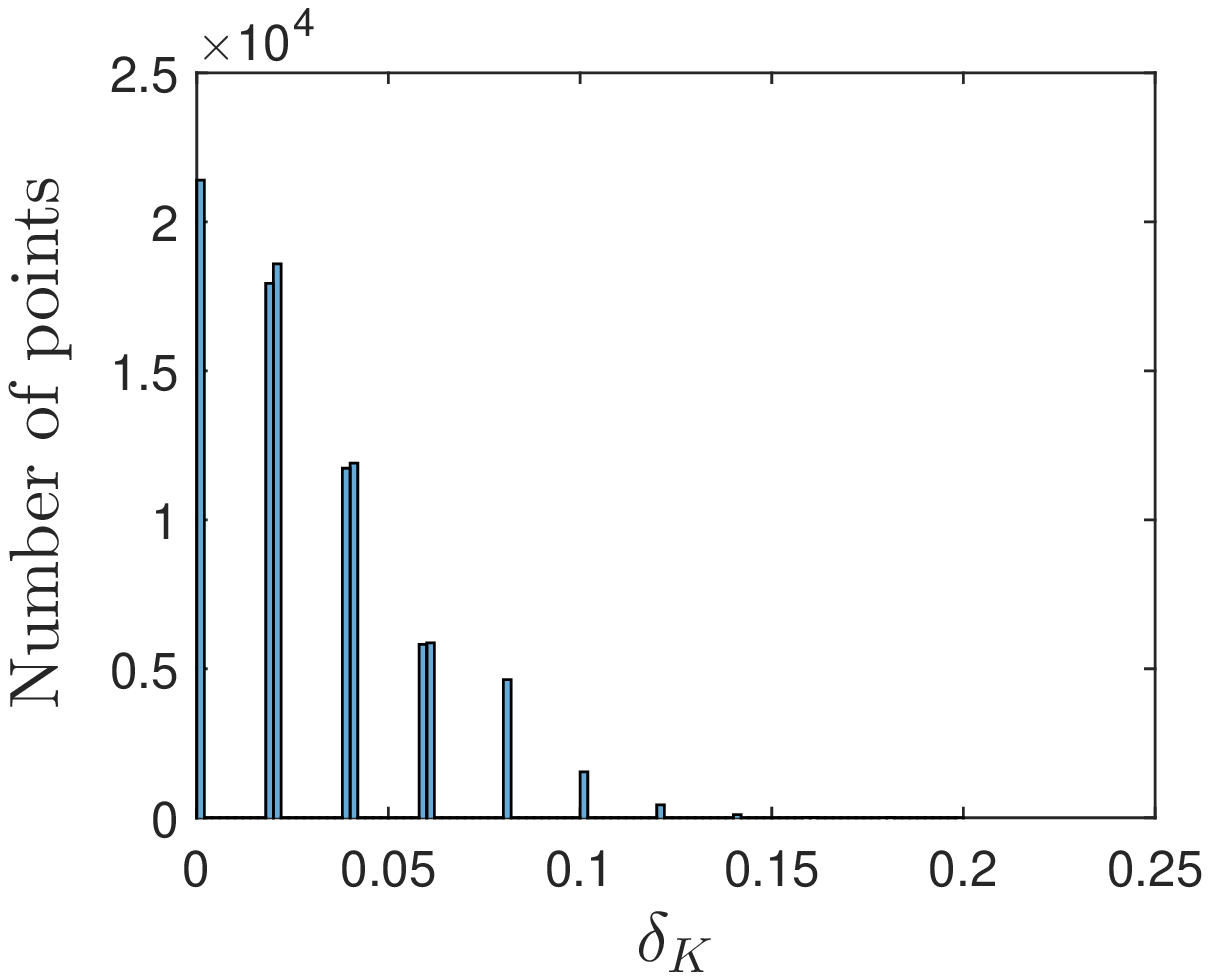,width=1.7in}}
  \centerline{(b)}
\end{minipage}
\begin{minipage}[b]{.5\linewidth}
  \centering{\epsfig{figure=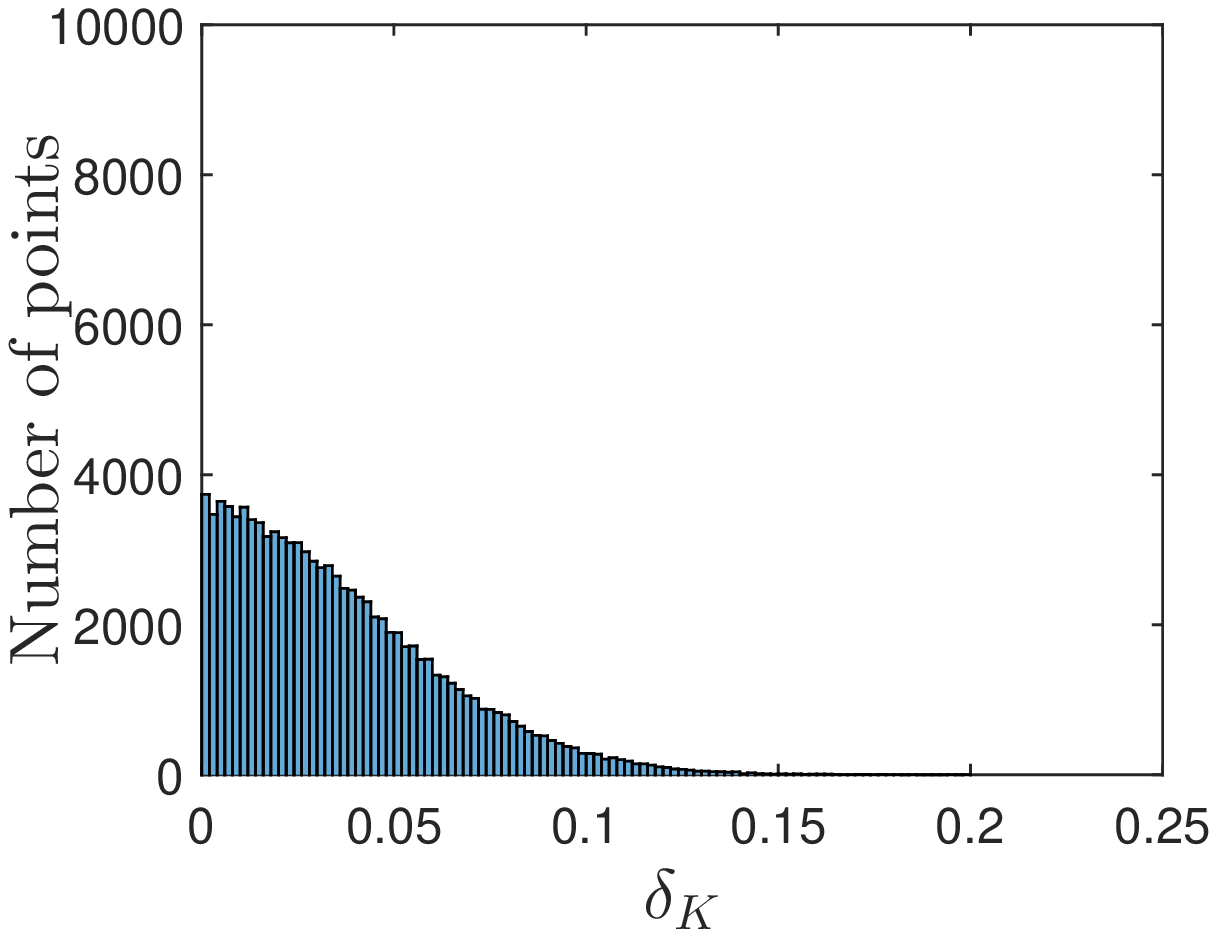,width=1.7in}}
  \centerline{(c)}
\end{minipage}
\begin{minipage}[b]{.48\linewidth}
  \centering{\epsfig{figure=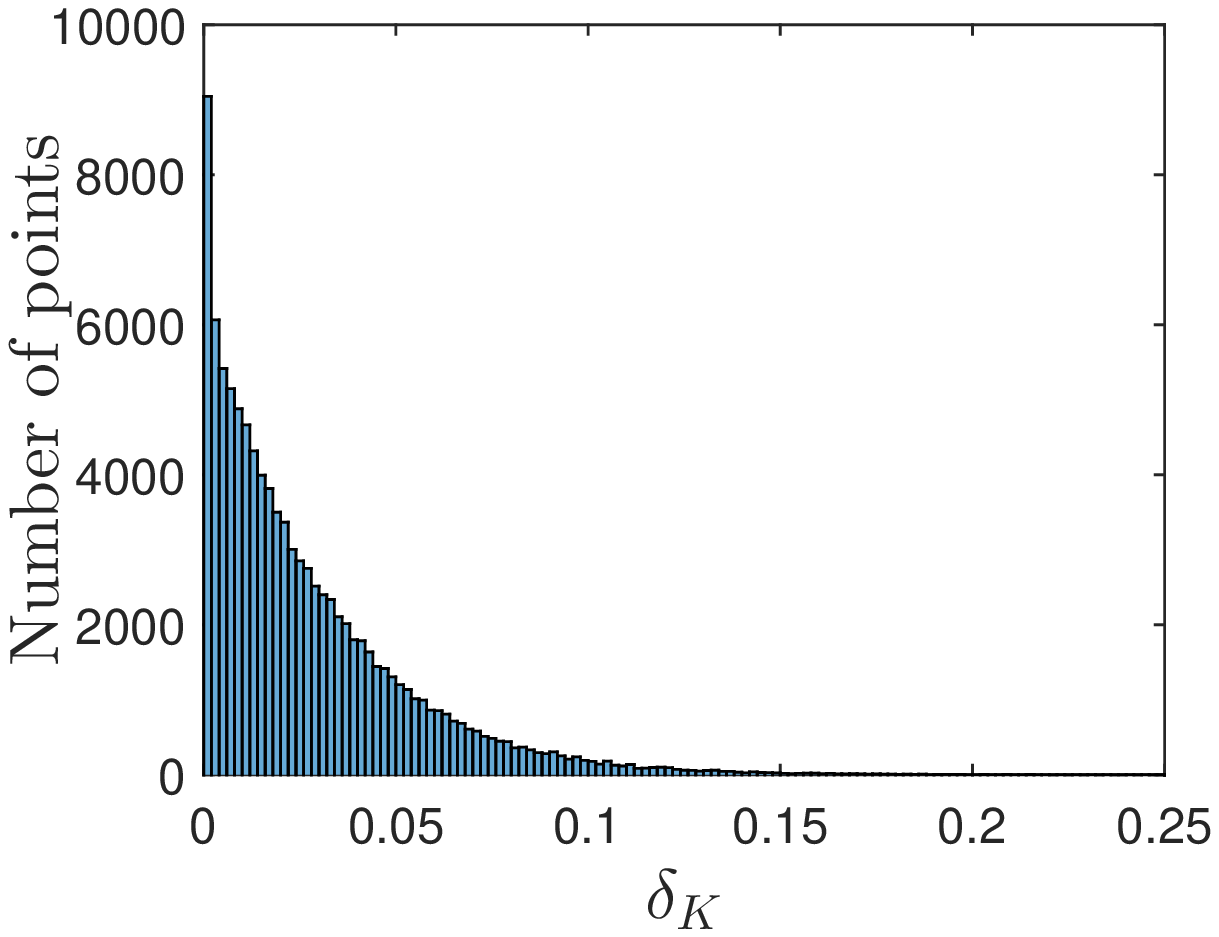,width=1.7in}}
  \centerline{(d)}
\end{minipage}
\caption{Histogram for density estimation of $\delta_{K}$ with $N=4000$, $M=1000$, and $K=100$.
(a)(b) $\bm{x} \in \left\{0,1 \right\}^{N}$. (c)(d) $\bm{x}$ is drawn from $\mathcal{N}(0,1)$.  (a)(c) Gaussian random matrix. (b)(d) The proposed matrix. }
\label{fig:RIP Histogram}
\end{figure}

\begin{table}[!htbp]
\caption{Estimation of $\delta_{K}$ for $\bm{A}$ being either a Gaussian random matrix or the proposed projection matrix under $N=4000$ and different $M$ and $K$.
The result is presented by $a/b$, where $a$ and $b$ denote $E\left\{ \delta_{K} \right\}$'s obtained by Gaussian random matrix and the proposed matrix, respectively.
Bold represents the better results.}
\footnotesize
\begin{tabular}{|c|c|c|c|c|c|}
\hline
\backslashbox{M}{K}& $25$ & $50$ & $100$ & $200$ & $400$ \\
\hline
1000 & .035/\textbf{.015} & .036/\textbf{.025} & .036/\textbf{.028} & .0.37/\textbf{.030} & .037/\textbf{.031} \\
500 & .048/\textbf{.031} & .049/\textbf{.040} & .050/\textbf{.044} & .050/\textbf{.045} & .051/\textbf{.046} \\
250 & .070/\textbf{.055} & .070/\textbf{.065} & .069/\textbf{.066} & .071/\textbf{.067} & .073/\textbf{.069} \\
125 & .101/\textbf{.093} & .100/\textbf{.095} & .105/\textbf{.097} & .101/\textbf{.098} & .101/\textbf{.100} \\
\hline
\end{tabular}
\normalsize
\label{table: RIC for different matrices}
\end{table}

\begin{table}[!htbp]
\caption{Estimation of $\delta_{K}$ for $\bm{A}$ being either a Gaussian random matrix or the proposed projection matrix.
Except that the non-zero entries of $\bm{x}$ are drawn from $\mathcal{N}(0,1)$, other settings follow Table \ref{table: RIC for different matrices}.  }
\footnotesize
\begin{tabular}{|c|c|c|c|c|c|}
\hline
\backslashbox{M}{K}& $63$ & $125$ & $250$ & $500$ & $1000$ \\
\hline
1000 & .033/\textbf{.011} & .035/\textbf{.018} & .037/\textbf{.026} & .0.36/\textbf{.028} & .036/\textbf{.030} \\
500 & .048/\textbf{.025} & .050/\textbf{.035} & .050/\textbf{.039} & .050/\textbf{.042} & .051/\textbf{.044} \\
250 & .070/\textbf{.047} & .068/\textbf{.061} & .069/\textbf{.063} & .072/\textbf{.065} & .072/\textbf{.068} \\
125 & .101/\textbf{.080} & .102/\textbf{.091} & .101/\textbf{.094} & .101/\textbf{.096} & .101/\textbf{.97} \\
\hline
\end{tabular}
\normalsize
\label{table: RIC for different matrices with gaussian signal}
\end{table}

%We note that if the range of $\bm{x}$ belongs to the finite set ({\em i.e.} $\bm{x}\in \left\{ -1,1 \right\}^{N}$), because all entries of the proposed matrix are one, the range of $\bm{A}\bm{x}$ is also finite and maybe includes zero. When $(\bm{A}\bm{x})_{i}=0$, $sign((\bm{A}\bm{x})_{i})$ is ambiguous to be assigned to be 1 or -1. In this case, the performance is possible to degrade. To avoid the problem, the pre-processing for $\bm{x}$, such as subtracting $\bm{x}$ by the mean in \cite{Gong2013}, is efficient.

\section{Experimental Results}\label{sec:exper}
Simulations were conducted in Matlab environment with an Intel CPU Q6600 and $16$ GB RAM under MS Win7 ($64$ bits).
Since we focus on the comparison of computation and storage costs, we only compare the proposed algorithm with some selected data-independent binary embedding algorithms, including
\begin{enumerate}
\item [$\bullet$] Locality Sensitive Hashing (LSH) \cite{Charikar2002}: $A$ is a Gaussian random matrix. This method is considered as a baseline in terms of performance and computation cost.
\item [$\bullet$] CBE-rand \cite{ChangeFu2014}: $A$ is designed as a circulant matrix, where the seed vector is a Gaussian random vector. This method focuses on fast embedding by FFT.
\item [$\bullet$] BP-rand \cite{Gong2013}: Use two matrices to separably project data.
We follow the data-independent setting in \cite{Gong2013}, where two matrices are designed as Gaussian random matrices without learning.
\end{enumerate}

Except the proposed method and BP-rand, all other codes were downloaded from http://www.unc.edu/~yunchao/.
%LSH, CBE-rand, BP-rand and the proposed method are data-independent approaches.
%\textcolor{red}{Since we focus on the comparison of computation and storage costs, data-dependent techniques such as ITQ are ignored and must be equal to or slower than LSH (with/without considering learning). But, data-dependent techniques are expected to have better performance. }
According to the following evaluations, our method is concluded to be very efficient to compute binary codes with low memory requirements and exhibit performance of image classification and retrieval being comparable to state-of-the-art data-independent projection techniques.

\begin{figure}[h]
\begin{minipage}[b]{.5\linewidth}
  \centering{\epsfig{figure=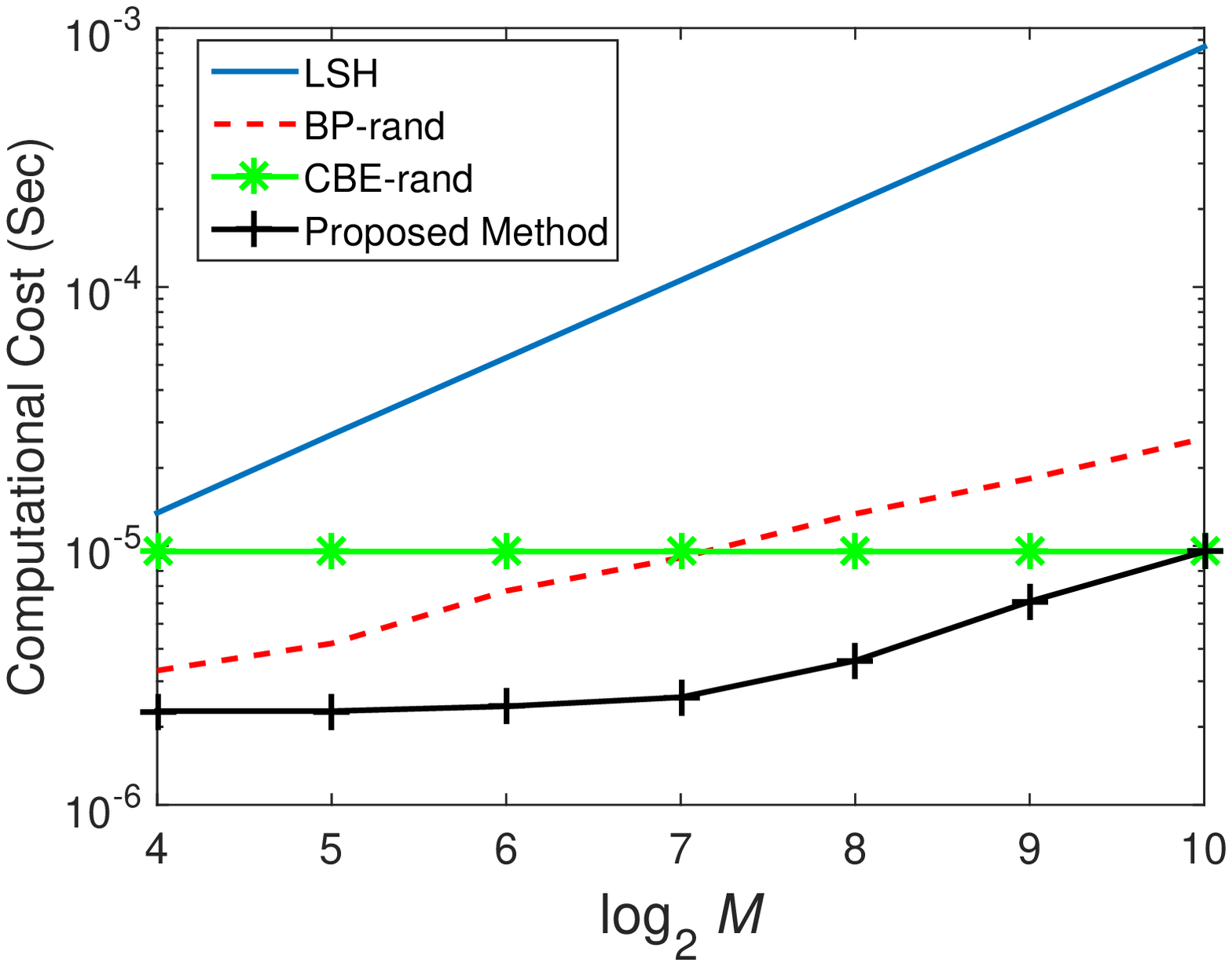,width=1.7in}}
  \centerline{(a)}
\end{minipage}
\begin{minipage}[b]{.48\linewidth}
  \centering{\epsfig{figure=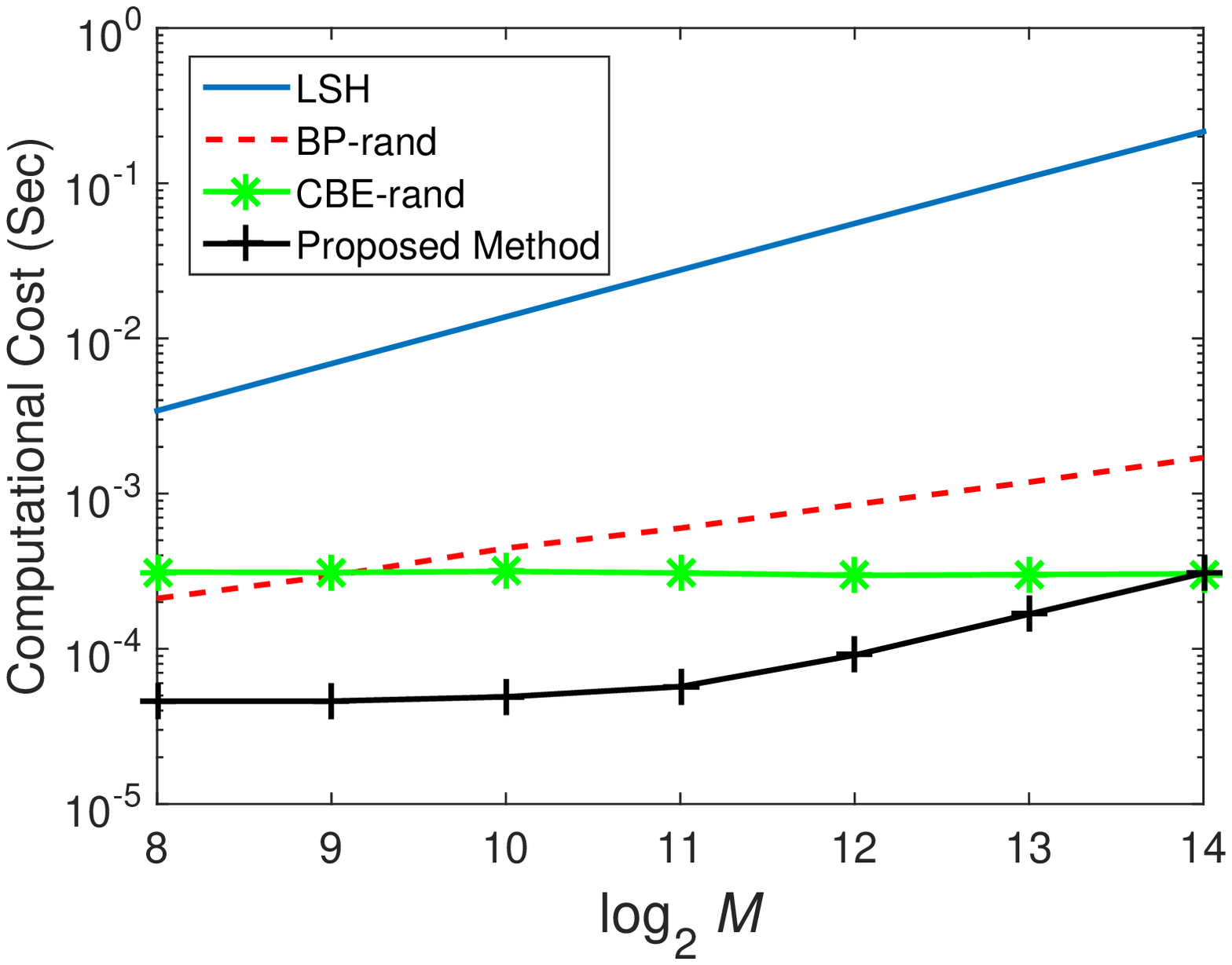,width=1.7in}}
  \centerline{(b)}
\end{minipage}
\begin{minipage}[b]{.5\linewidth}
  \centering{\epsfig{figure=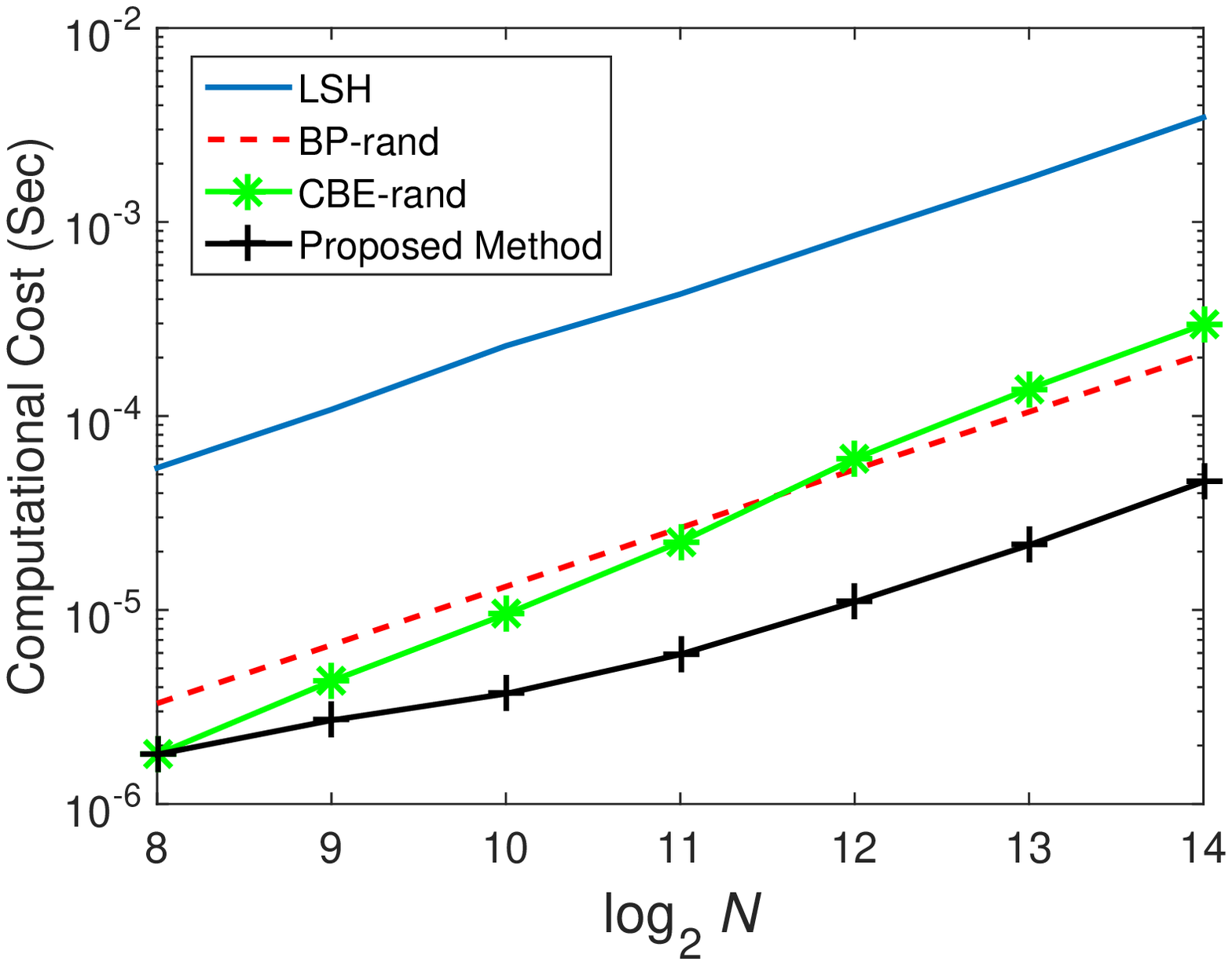,width=1.7in}}
  \centerline{(c)}
\end{minipage}
\begin{minipage}[b]{.48\linewidth}
  \centering{\epsfig{figure=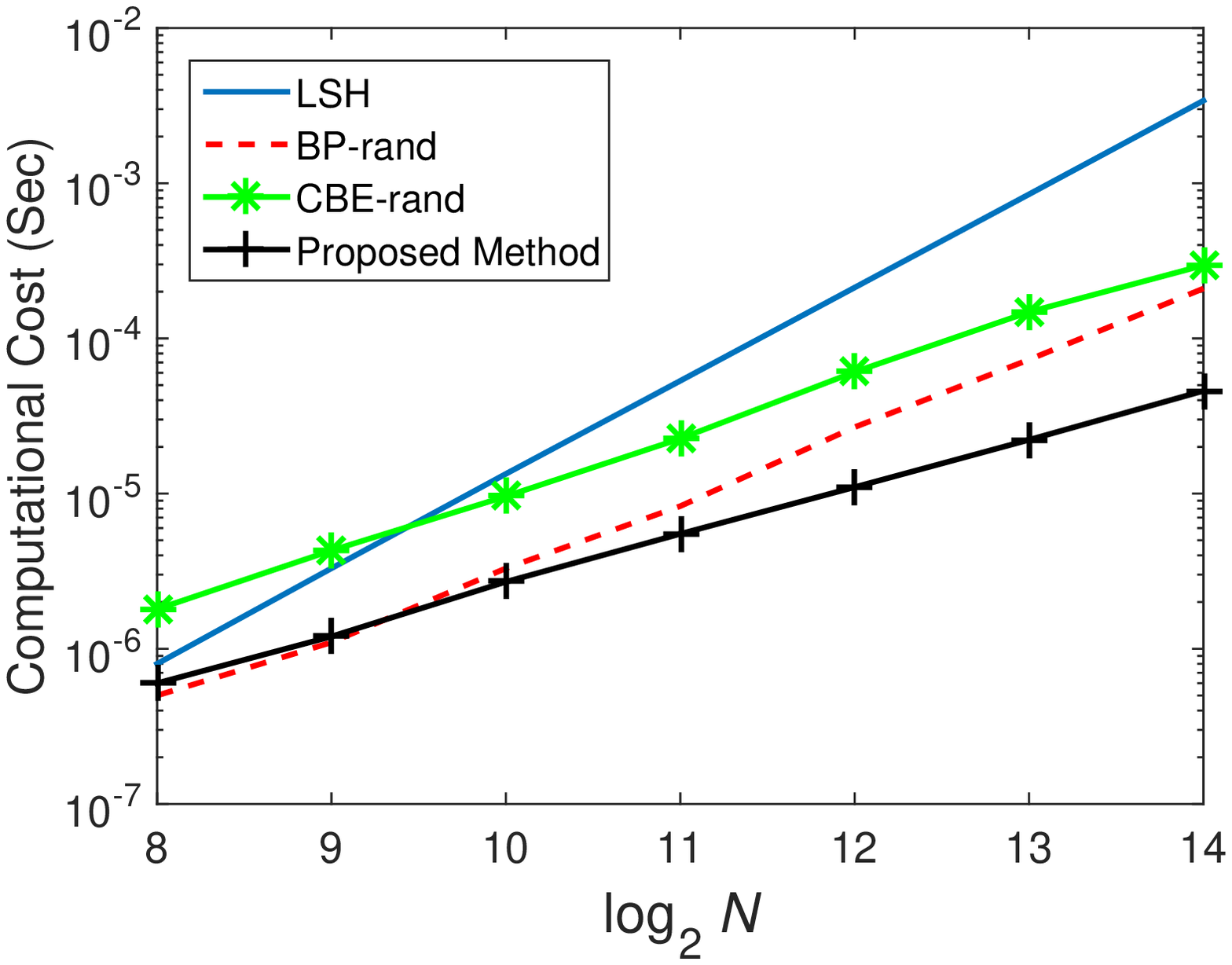,width=1.7in}}
  \centerline{(d)}
\end{minipage}
\caption{Comparisons between different approaches in terms of computation cost. (a) Fixing $N=2^{10}$, $M$ versus computation time.  (b) Fixing $N=2^{14}$, $M$ versus computation time. (c) Fixing $M=2^8$, $N$ versus computation time. (d) Fixing the compression ratio $\frac{N}{M}=2^6$, $N$ versus computation time.  }
\label{fig:Computation cost}
\end{figure}

\subsection{Computation and Memory Costs}\label{sec:exper 1}
%We verify the computation cost versus different parameters and storage cost.
Since computation cost are invariant to signal types, synthesis data were used here.
%Then, input data $\bm{x}$ is fed into algorithms to estimate the computation cost (sec) of embedding process.
Storage cost is equal to the memory requirement for saving projection $\bm{A}$.
%ITQ is ignored since it is data-dependent and must be equal to or slower than LSH (with/without considering learning).
Figs. \ref{fig:Computation cost}(a) and (b) show the computation time versus different $M$'s under $N=2^{10}$ and $N=2^{14}$, respectively.
Fig. \ref{fig:Computation cost}(c) shows the results obtained from different $N$'s under $M=2^{8}$.
One can clearly find that the proposed method outperforms the other methods (note the logarithmic scale of the vertical axis).
We can validate the experimental results along with theoretical results in Table \ref{table: computation cost and storage cost}.
When $M=N$, the computation cost of our method is equal to that of CBE-rand.
When $M<N$, we have two observations from Figs. \ref{fig:Computation cost}(a)$\sim$(c):
(i) our method is dominated by $O(N)$ when $M \log M < N$  and
(ii) $O(M \log M)$ dominates the computation cost when $M \log M \geq N$.
%In sum, no matter whichever slope is outperforms other methods.

In addition, fixing the compression ratio $\frac{N}{M}$, we have $M \log M > N$ for sufficiently large $N$.
It implies that for high-dimensional signals with a fixed compression ratio, the proposed method speeds up projection remarkably.
Fig. \ref{fig:Computation cost}(d) further shows the computation cost of our method increases slower than other methods with constant compression ratio.
It should be noted that BP-rand outperforms CBE-rand and our method when $N\leq 2^9$ because
(i) $N^{1.5}$ approximates $N \log N$ when $N$ is small and
(ii) CBE-rand and our method incur larger Big-O constants due to the use of FFT.

On the other hand, Table \ref{table: storage cost in MB} shows the comparison of memory cost for saving projection.
We follow the parameter setting in \cite{Gong2013} with $M=N$.
It is observed that our method is nearly comparable to CBE-rand.

However, our method actually requires less memory and outperforms CBE-rand under practical scenario with $M < N$, as depicted in Table \ref{table: storage cost in MB with different M}.
This is because the cost of $\bm{\Phi}$ in our method only depends on $M$ but that in CBE-rand depends on $N$.

\begin{table}[!htbp]
\caption{Memory (MegaBytes) needed to store the projection matrix, assuming each element is float-point ($32$ bits).
%The parameter setting follows \cite{Gong2013}
Note the results of BP-rand is directly copied from Table 2 of \cite{Gong2013}.}
\begin{tabular}{|c|c|c|c|c|}
\hline
$N$ & LSH  & BP-rand  & CBE-rand  &  Ours  \\
\hline\hline
$1.28\times 10^3$ & $6.25$ & $0.06$ & $0.0049$ & $0.005$   \\
$1.28\times 10^4$ & $625$ & $0.10$  & $0.049$ & $0.0504$    \\
$2.56\times 10^4$ & $2500$ & $0.22$ & $0.0977$ & $0.1007$    \\
$6.4 \times 10^4$ & $15625$ & $1.02$ & $0.2441$ & $0.2518$ \\
$1.28\times 10^5$ & $62500$ & $3.88$  & $0.4883$ & $0.5035$  \\
\hline
\end{tabular}
\label{table: storage cost in MB}
\end{table}

\vspace{-10pt}
\begin{table}[!htbp]
\caption{Memory (MegaBytes) required for our method and CBE-rand under fixed $N=1.28\times 10^5$ and various $M$.}
\small
\begin{tabular}{|c||c|c|c|c|}
\hline
$M$ &   $1.6\times 10^4$ & $3.4\times 10^4$ &  $6.4\times 10^4$ & $1.28\times 10^5$   \\
\hline
 CBE-rand  & 0.4883 & 0.4883  & 0.4883 & 0.4883   \\
Ours & 0.0763 & 0.1373  & 0.2594 & 0.5035    \\
\hline
\end{tabular}
\normalsize
\label{table: storage cost in MB with different M}
\end{table}

\subsection{Image Applications}\label{sec:exper 2}
We verify whether binary codes yielded after our embedding scheme, despite its low computation and storage cost, still contain discriminative power in image classification and retrieval.

\subsubsection{Image Classification}
Two datasets were considered in image classification:
\begin{enumerate}
\item [$\bullet$]  CIFAR \cite{Krizhevsky2009}: It consists of $64,800$ images that have been manually grouped into $11$ ground-truth classes (airplane, automobile, bird, boat, cat, deer, dog, frog, horse, ship and truck). All images were represented as GIST descriptor \cite{Oliva2001} with $N=2048$.
\item [$\bullet$]  MNIST \cite{LeCun1998}: It includes $60,000$ images with handwriting digits from $0-9$.
All images were represented as GIST descriptor \cite{Oliva2001} with $N=512$.
\end{enumerate}

After embedding, binary codes were fed into LIBSVM \cite{Chang-SVM2011} to train classifier by supervised learning (8-fold cross-validation).
Ground truth is based on pre-defined labels provided by the datasets.
Fig. \ref{fig:Accuracy by SVM} shows the accuracy versus different $M$ bits, where accuracy is the probability that classifier has labeled an testing image into the ground truth.
In both CIFAR and MNIST datasets, the proposed method is comparable to LSH and CBE-rand, but the performance of BP-rand degrades due to projection within a bilinear structure.
In addition, though GIST feature is not sparse, our method still exhibits good performance because the features are still approximately sparse, where only few entries are significant.

\begin{figure}[h]
\begin{minipage}[b]{.5\linewidth}
  \centering{\epsfig{figure=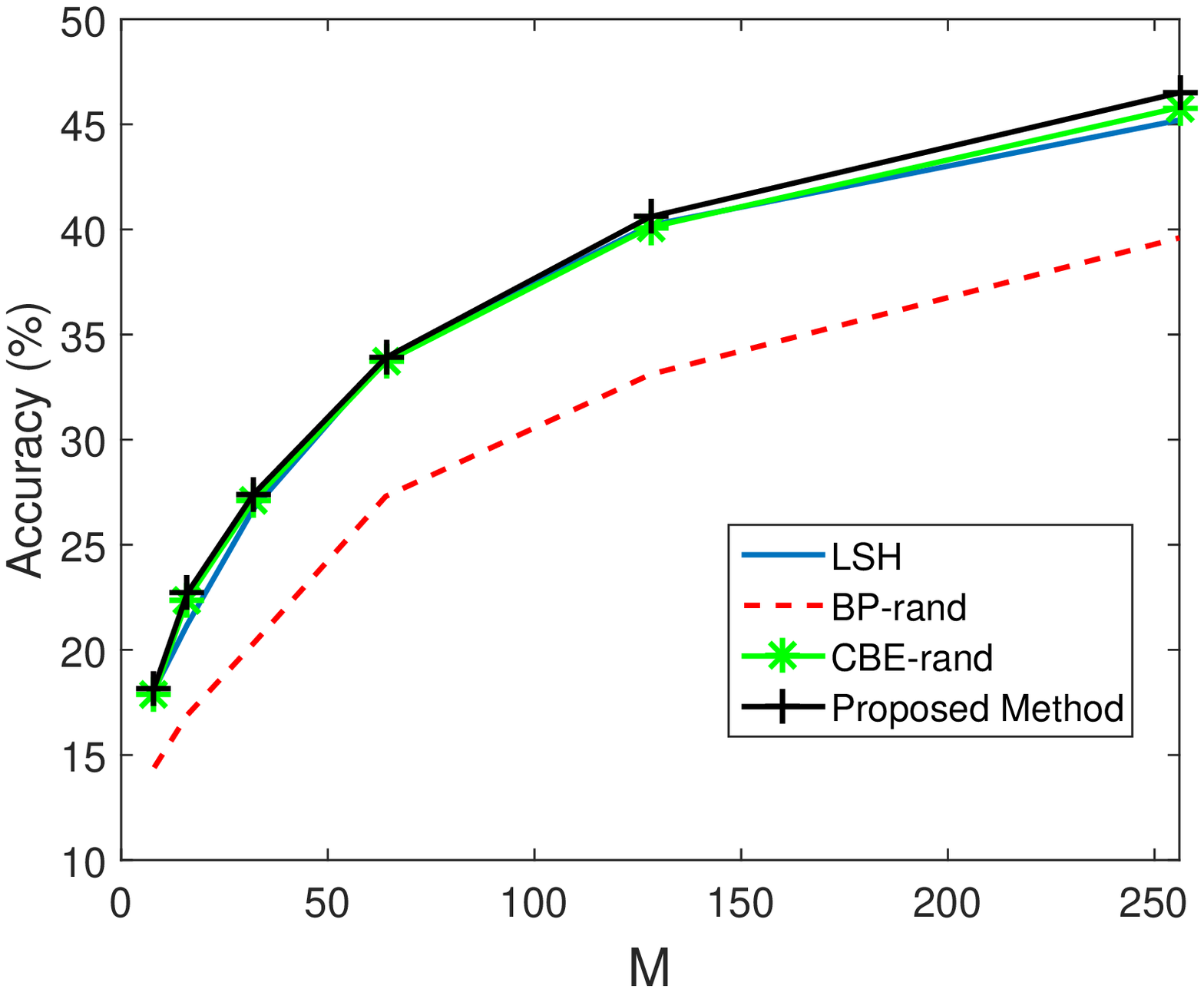,width=1.9in}}
  \centerline{(a)}
\end{minipage}
\begin{minipage}[b]{.48\linewidth}
  \centering{\epsfig{figure=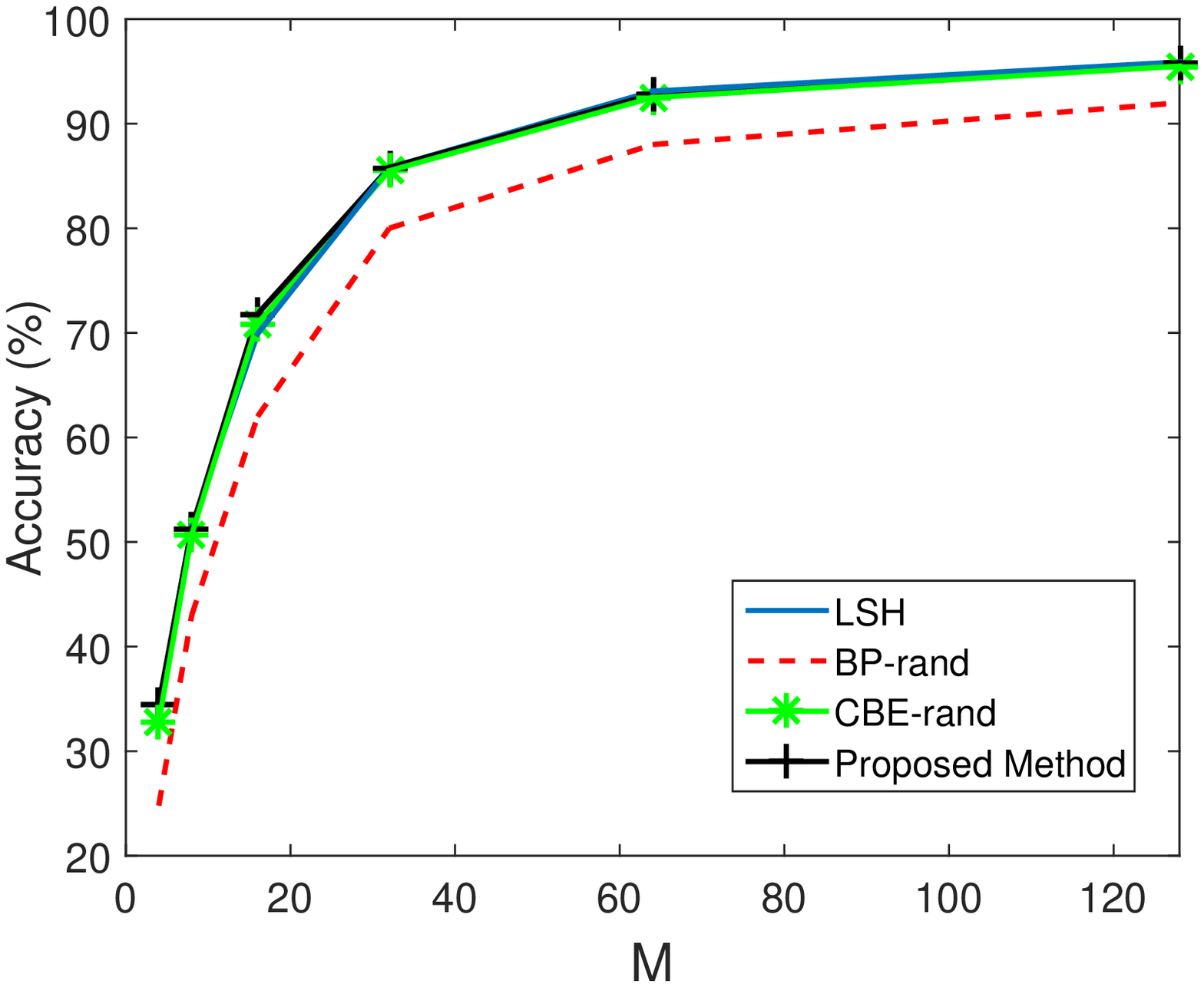,width=1.9in}}
  \centerline{(b)}
\end{minipage}
\caption{Accuracy vs. $M$. Classifier is learned by LIBSVM. (a) CIFAR dataset  (b) MNIST dataset.  }
\label{fig:Accuracy by SVM}
\end{figure}

\subsubsection{Image Retrieval }\label{sec:exper 2}
For purpose of image retrieval, we used the same datasets and setting in image classification.
All images still were represented by GIST features.
%Ground truth is based on pre-defined labels provided by the datasets.
In this experiment, ``retrieval'' was performed by randomly selecting $1,000$ query images from dataset and returning images according to hamming distance sorting in an ascending order.
Performance is measured by mean Average Precision (mAP) \cite{Gong2013-ITQ}.

Fig. \ref{fig:Retrival versus M with fixed top} shows mAP with top $50$ returned images.
Whatever $M$ is, the proposed approach has the comparable performance with LSH and CBE-rand.
%It means that hamming distance are used as a similarity measure in proposed binary codes to efficiently search similar images.
In other words, the proposed method preserves angle (similarity) well even an extra downsampling matrix is introduced to achieve faster binary embedding.

%Fig. \ref{fig:Retrival versus fixed M with different top} shows that mAP decreases along with the increase of the number of returned images. Each class in both CIFAR and MNIST datasets includes about 6000 images. Ideally, the performance should be invariant until returned images exceeds 6000. But, it is still a challenge for existing methods.

\begin{figure}[h]
\begin{minipage}[b]{.5\linewidth}
  \centering{\epsfig{figure=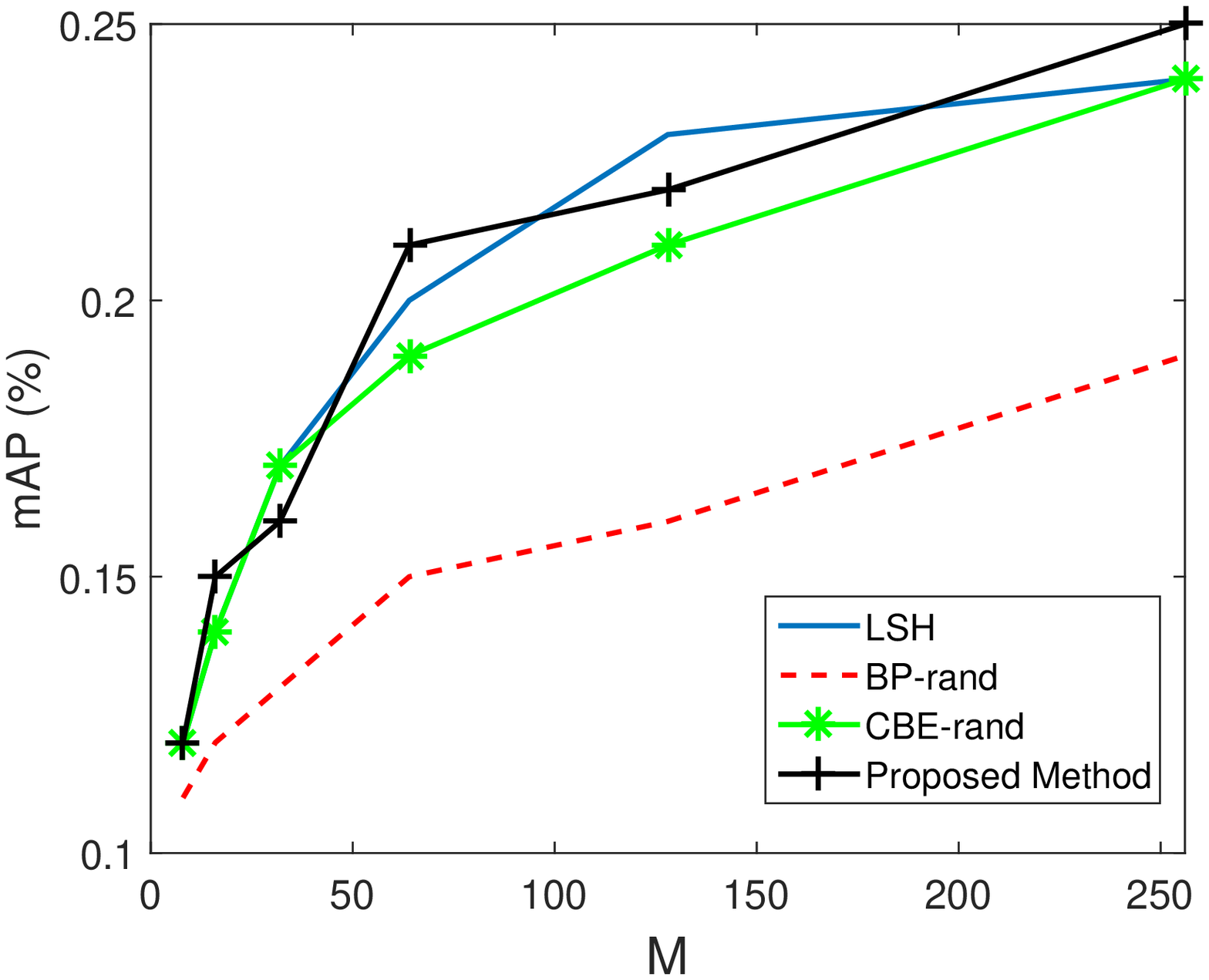,width=1.9in}}
  \centerline{(a)}
\end{minipage}
\begin{minipage}[b]{.48\linewidth}
  \centering{\epsfig{figure=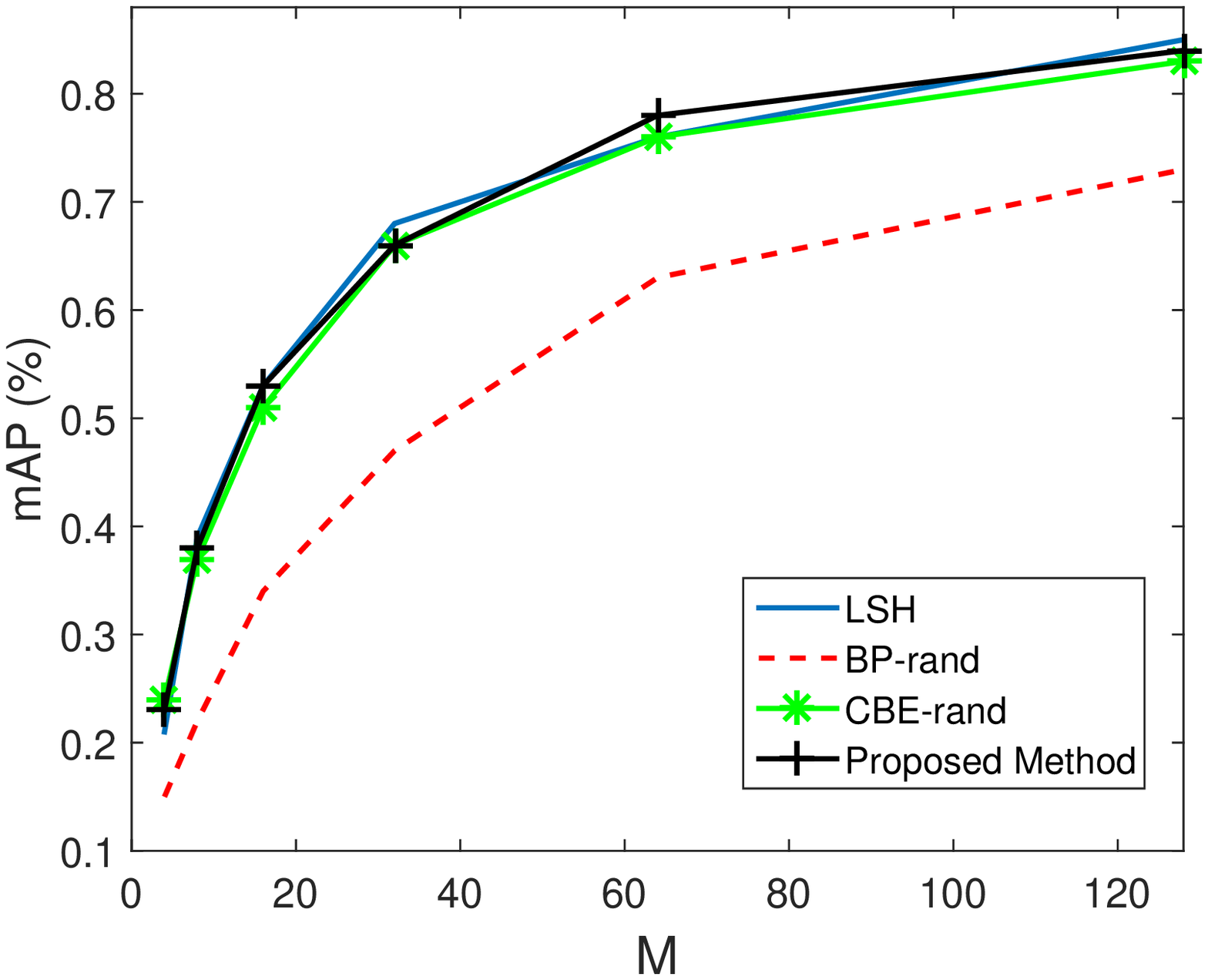,width=1.9in}}
  \centerline{(b)}
\end{minipage}
\caption{mAP vs. $M$ with top $50$ returned images. (a) CIFAR dataset.  (b) MNIST dataset.  }
\label{fig:Retrival versus M with fixed top}
\end{figure}

%\begin{figure}[h]
%\begin{minipage}[b]{.5\linewidth}
%  \centering{\epsfig{figure=Retrieval_CIFAR_M.eps,width=1.9in}}
%  \centerline{(a)}
%\end{minipage}
%\begin{minipage}[b]{.48\linewidth}
%  \centering{\epsfig{figure=Retrieval_MNIST_M.eps,width=1.9in}}
%  \centerline{(b)}
%\end{minipage}
%\caption{mAP versus the number of returned images. (a) CIFAR dataset with $M=64$.  (b) MNIST dataset with $M=32$.  }
%\label{fig:Retrival versus fixed M with different top}
%\end{figure}

\section{Conclusions and Future Works}\label{sec:cons}
In this paper, we have proposed a data-independent binary embedding technique with $O(N+M\log M)$ in computation cost and $O(N)$ in storage cost to outperform state-of-the-art approaches.
We also theoretically prove that if data have sparsity, similarity (angle) between data is preserved well.
The full potential of our method is applied for ultra-high dimensional data \cite{ChangeFu2014}, for which no other methods are applicable.

For future work, the goal is to extend our method to data-dependent paradigm.
That is, given $\bm{R}$, $\bm{\Phi}\bm{R}\bm{x}$ is considered to be new training data instead of $\bm{x}$. All we need to do is to learn a circulant matrix $\bm{D}$.
Thus, the learning process applies to low-dimensional data ($\bm{\Phi}\bm{R}\bm{x}$), resulting in low computation and memory costs.
After that, our goal is to simultaneously learn $\bm{D}$ and $\bm{R}$.

%\begin{figure}[h]
%\begin{minipage}[b]{1\linewidth}
%  \centering{\epsfig{figure=E_HammingDis.eps,width=2in}}
%\end{minipage}
%\caption{The illustration of angle-preserving prepuberty for the proposed matrix, where $N=1024$ and $M=64,128,256,512$ are labeled as ''Proposed-64``,...,''Proposed-512`` respectively. In addition, ''Rand`` represents the outcome by Gaussian random matrix.  }
%\label{fig:EValue of Hamm}
%\end{figure}
%
%
%\begin{figure}[h]
%\begin{minipage}[b]{.5\linewidth}
%  \centering{\epsfig{figure=Var_Hamming_12.eps,width=1.6in}}
%  \centerline{(a) $\displaystyle \theta=\pi/12$}
%\end{minipage}
%\begin{minipage}[b]{.48\linewidth}
%  \centering{\epsfig{figure=Var_Hamming_6.eps,width=1.6in}}
%  \centerline{(b) $\displaystyle \theta=\pi/6$}
%\end{minipage}
%\begin{minipage}[b]{.5\linewidth}
%  \centering{\epsfig{figure=Var_Hamming_3.eps,width=1.6in}}
%  \centerline{(c) $\displaystyle \theta=\pi/3$}
%\end{minipage}
%\begin{minipage}[b]{.48\linewidth}
%  \centering{\epsfig{figure=Var_Hamming_2.eps,width=1.6in}}
%  \centerline{(d) $\displaystyle \theta=\pi/2$}
%\end{minipage}
%\caption{$Var\left\{ \mathcal{H}_{M}\left(\bm{x}_{1},\bm{x}_{2}\right) \right\}$ versus different $M$'s under (a) $\theta=\frac{\pi}{12}$ (b) $\theta=\frac{\pi}{6}$ (c) $\theta=\frac{\pi}{3}$ (d) $\theta=\frac{\pi}{2}$. The proposed matrix obtains smaller variance compared with Gaussian random matrix.}
%\label{fig:Variance of Hamm}
%\end{figure}

\section{Acknowledgment}
This work was supported by Ministry of Science and Technology, Taiwan, ROC, under grants MOST 104-2221-E-001-019-MY3 and 104-2221-E-001-030-MY3.

\bibliographystyle{IEEEbib}	% (uses file "IEEEbib.bst")
\bibliography{refs}		% expects file "refs.bib"

\end{document}